\documentclass[journal]{IEEEtran}

\usepackage{cite}
\usepackage{amsmath,amssymb,amsfonts}
\usepackage{algorithmic}
\usepackage{algorithm}
\usepackage{graphicx}
\usepackage{textcomp}
\usepackage{xcolor}
\usepackage{booktabs}
\usepackage{multirow}
\usepackage{siunitx}
\usepackage{enumitem}
\usepackage{makecell}
\usepackage[caption=false,font=footnotesize]{subfig}
\usepackage{url}
\usepackage[hidelinks]{hyperref}
\usepackage{amsmath, amssymb, amsthm}
\usepackage{amsmath} 
\usepackage{amssymb}  

\usepackage{todonotes} 

\usepackage{algorithm}
\usepackage{algorithmic}
\allowdisplaybreaks

\theoremstyle{plain}
\newtheorem{theorem}{Theorem}

\theoremstyle{definition}
\newtheorem{assumption}{Assumption}

\theoremstyle{remark}
\newtheorem{remark}{Remark}

\title{\LARGE \bf
End-to-End Battery Dispatch with Exact Rainflow Degradation via Mixed-Integer Differentiable Predictive Control
}
\author{
Eshagh Safarzadeh Ravajiri, 
J\'an Drgo\v na, 
Mahdi Mehrtash,
and Benjamin F. Hobbs%

\thanks{Code available at \url{https://github.com/SOLARIS-JHU/Battery-MI-DPC.git}.}%

\thanks{E.~Safarzadeh Ravajiri and B.~F.~Hobbs are with the Department of Environmental Health and Engineering, The Johns Hopkins University, Baltimore, MD 21218, USA (e-mail: esafarz1@jh.edu; bhobbs@jhu.edu).
J.~Drgo\v{n}a is with the Department of Civil and Systems Engineering, The Johns Hopkins University, Baltimore, MD 21218, USA (e-mail: jdrgona1@jh.edu).
M.~Mehrtash is with the Department of Electrical and Biomedical Engineering, University of Nevada, Reno, NV 89557, USA (e-mail: mahdi.mehrtash@ieee.org).
}
\thanks{(Corresponding authors: E.~Safarzadeh Ravajiri; M.~Mehrtash.)}
}

\begin{document}
\maketitle


\begin{abstract}
Optimal dispatch of battery energy storage systems requires balancing energy arbitrage against cycle-induced degradation, which is accurately quantified through rainflow cycle counting. However, rainflow's combinatorial, nondifferentiable algorithm is incompatible with both convex optimization and gradient-based neural network training. We present a self-supervised mixed-integer differentiable predictive control framework that trains neural policies directly on exact rainflow degradation through a novel differentiable rainflow layer combining exact gradients at state-of-charge extrema with dense proxy gradients on incremental changes, enabling stable end-to-end training while preserving true degradation physics. A mixed-integer differentiable architecture enforces power balance, dynamics, and mode exclusivity, with a safety filter guaranteeing constraint satisfaction. We evaluate the framework on 3,650 real battery-day scenarios (a 10-battery fleet over 365 days) spanning three utility regions (SDG\&E California, Xcel Energy Colorado, APS Arizona) with diverse time-of-use pricing structures. A single-battery trained policy achieves a 0.35\% performance gap on its training distribution with over 200$\times$ speedup, while fleet-wide training generalizes across all households and utility regions, achieving a 3.33\% performance gap with 564$\times$ computational speedup (62 seconds vs.\ 9.7 hours) and 100\% feasibility. The millisecond-scale inference reduces computation by over two orders of magnitude compared to mixed-integer solvers, enabling practical deployment at scale.
\end{abstract}



\section{Introduction}
\IEEEPARstart{B}{attery} energy storage systems (BESS) are emerging as key flexibility resources in modern power systems, enabling energy arbitrage, peak shaving, frequency regulation, and renewable integration \cite{sakib2025_role_bess_rez}. Effective BESS dispatch must exploit price signals while maintaining reliability and ensuring long-term economic value under realistic degradation and warranty constraints \cite{abdulla2018_optimal_ess_sdp}, placing emphasis on control methods that translate forecasts of load, generation, and prices into economically efficient, physically feasible charging and discharging schedules.
A defining feature of BESS is that economic performance is tightly coupled to degradation. Lithium-ion batteries exhibit capacity fade, often dominated by cycle aging in actively operated systems \cite{lam2025_calendar_aging_joule}. Cycle-induced degradation depends nonlinearly on cycle depth where deeper cycles cause disproportionately greater damage than shallower cycles with equivalent energy throughput \cite{xu2018_liion_degradation_model}. Optimal dispatch must therefore manage battery cycles, as naively optimizing for short-term arbitrage can lead to aggressive cycling, resulting in rapid erosion of battery capacity \cite{wankmueller2017_degradation_arbitrage}.

Degradation-aware models decompose the state of charge (SOC) trajectory into equivalent cycles using depth-dependent stress functions calibrated from laboratory data \cite{collath2022_aging_aware_operation_review}. The industry-standard tool is rainflow counting, developed for structural fatigue analysis \cite{astm2023_e1049} and now standard in battery research \cite{endo1968_rainflow,vetter2005_ageing_liion_full}. The algorithm maintains a stack-based memory of SOC peaks and valleys, iteratively closing cycles as new extrema are encountered \cite{xu2018_liion_degradation_model}. This process is sequential, combinatorial, and nondifferentiable, incompatible with convex optimization and gradient-based learning. Combined with binary mode variables and bilinear power flow terms, this yields a mixed-integer nonlinear program (MINLP) requiring specialized solvers that may take hours or fail to converge.

In parallel, learning-based approaches have emerged as promising alternatives to classical optimization. Learning-to-optimize (L2O) methods train neural networks to approximate the solution map of parametric optimization problems, generating near-optimal decisions in milliseconds with computational speedups of several orders of magnitude~\cite{khaloie2025_ml_opf_review,singh2022_sensitivity_dnn_acopf,chen2024_e2e_ed}. A complementary direction is Differentiable Predictive Control (DPC)~\cite{drgona2024dpc_guarantees}, which embeds system dynamics and operational constraints directly into the neural architecture, enabling self-supervised end-to-end training without requiring pre-solved optimal datasets. Recent works has extended DPC to mixed-integer settings~\cite{boldocky2025mixedintegerdpc,boldocky2025dpc}.

Learning-to-optimize and end-to-end differentiable-optimization methods have been applied to economic dispatch and storage scheduling, and have even been equipped with differentiable feasibility-repair layers and self-supervised training \cite{chen2024_e2e_ed}. Despite these advances, learning-based methods have largely avoided the core challenge of degradation-aware dispatch. Existing approaches either incorporate degradation via linear throughput penalties that do not capture nonlinear cycle-depth effects \cite{das2024adp_der_coordination}, approximate cycle-based degradation using simplified linearized models instead of full rainflow counting \cite{kwon2022rl_cycle_degradation}, or rely on reward shaping that guides the agent with surrogate degradation costs rather than true nonlinear degradation \cite{cao2020drl_battery_arbitrage}. In addition, reinforcement learning (RL) approaches to storage dispatch more broadly treat degradation as a black-box reward signal, discarding the known structure of the dispatch problem \cite{bui2020ddqn_bess, huang2021drl_pv_battery}, and struggle with strict constraint enforcement \cite{hou2023constraint_drl_ess}. The common obstacle across all of these directions is that the rainflow algorithm itself is nondifferentiable and combinatorial, which prevents its direct integration into the gradient-based training on which learning-based methods depend, and thereby precludes optimizing over true cycle-aware degradation.

This paper directly addresses this gap. We develop, to our knowledge, the first self-supervised framework that trains neural dispatch policies end-to-end on exact rainflow degradation. Our key contributions are as follows:

\begin{enumerate}[label=\arabic*.]
\item A differentiable rainflow layer that resolves the fundamental incompatibility between rainflow cycle counting and gradient-based learning.
\item A mixed-integer differentiable predictive control (MI-DPC) architecture that enforces power balance, dynamics, and mode exclusivity architecturally, with a convex safety projection filter whose well-posedness and trajectory feasibility are formally proven independent of network output quality.
\item Validation on a 10-battery fleet spanning three utility regions, SDG\&E (California), Xcel Energy (Colorado), and APS (Arizona), under diverse load and time-of-use pricing structures.
\end{enumerate}

Overall, the proposed framework provides a fast inference solver for the battery dispatch problem with degradation-aware optimization, significantly reducing the computational time required by traditional optimization approaches.

\section{Problem Formulation}
\label{sec:problem}

\subsection{Optimization Problem}

We formulate the residential BESS scheduling problem as a constrained optimization that minimizes total operational cost, comprising energy arbitrage and battery degradation, over a 24-hour planning horizon. The system consists of a grid-connected BESS co-located with a solar PV array under net metering. While formulated for day-ahead scheduling, this approach can naturally extends to receding horizon frameworks with periodic re-optimization at shorter intervals.

Given forecasted exogenous parameters $\xi = \{s_t, d_t, \pi^{\text{buy}}_t, \pi^{\text{sell}}_t\}_{t=0}^{T-1}$ representing solar generation [kW], electrical demand [kW], purchase price [\$/kWh], and sell price [\$/kWh] respectively, we determine the optimal battery charge/discharge schedule. The planning horizon $T = 24$ hours is discretized into $N$ time steps with uniform interval $\Delta t$ (15 minutes with $N = 96$).

The optimization problem is formulated as:
\vspace{-2mm}
\begin{equation}
\begin{aligned}
\min_{\mathbf{x}} \quad J(\mathbf{x}, \xi) = \sum_{t=0}^{T-1} &\left( p_t^{\text{import}} \pi_t^{\text{buy}} - p_t^{\text{export}} \pi_t^{\text{sell}} \right) \Delta t \\
&+ C_{\text{cyc}}(\text{SOC}_{0:T})
\end{aligned}
\label{eq:optimization_problem}
\end{equation}
\vspace{-2mm}

\vspace{-2mm}
\begin{subequations}
\label{eq:constraints}
\begingroup
\setlength{\jot}{8pt}
\begin{align}
\quad \text{s.t.} \quad
&z_t^{\text{ch}} + z_t^{\text{dis}} + z_t^{\text{idle}} = 1, && \forall t \label{eq:mode_ex} \\
&p_t^{\text{ch}} \leq P_{\text{max}} \cdot z_t^{\text{ch}}, && \forall t \label{eq:ch_couple} \\
&p_t^{\text{dis}} \leq P_{\text{max}} \cdot z_t^{\text{dis}}, && \forall t \label{eq:dis_couple} \\
&\text{SOC}_{t+1} = \text{SOC}_t + \left( \eta_{\text{ch}} p_t^{\text{ch}} - \frac{p_t^{\text{dis}}}{\eta_{\text{dis}}} \right) \Delta t, && \forall t \label{eq:soc_dyn} \\
&0.1 E_{\text{cap}} \leq \text{SOC}_t \leq 0.9 E_{\text{cap}}, && \forall t \label{eq:soc_lim} \\
&\text{SOC}_0 = \text{SOC}_T = 0.5 E_{\text{cap}} \label{eq:terminal} \\
&p_t^{\text{import}} - p_t^{\text{export}} = d_t - s_t + p_t^{\text{ch}} - p_t^{\text{dis}}, && \forall t \label{eq:power_bal}
\end{align}
\endgroup
\end{subequations}
\vspace{-4mm}

Decision variables include discrete operational modes $\mathbf{z}_t = [z_t^{\text{ch}}, z_t^{\text{dis}}, z_t^{\text{idle}}]^T \in \{0,1\}^3$, continuous battery power flows $p_t^{\text{ch}}, p_t^{\text{dis}} \in [0, P_{\text{max}}]$, state-of-charge $\text{SOC}_t \in [0.1 E_{\text{cap}}, 0.9 E_{\text{cap}}]$, and grid exchange $p_t^{\text{import}}, p_t^{\text{export}} \in \mathbb{R}_{\geq 0}$ at each time $t$.

This formulation constitutes a mixed-integer nonlinear program due to (i) discrete mode selection variables $\mathbf{z}_t$, (ii) bilinear coupling constraints~\eqref{eq:ch_couple}--\eqref{eq:dis_couple}, and (iii) the nonlinear, nondifferentiable degradation cost $C_{\text{cyc}}$. The optimal solution $\mathbf{x}^*$ balances energy arbitrage opportunities against battery degradation, yielding dispatch strategies that capture profitable price differentials while extending battery lifespan.

\subsection{Cycle Degradation Model}

Battery capacity fade during active operation is dominated by cycle aging. We model degradation using rainflow cycle counting combined with experimentally-derived stress functions. Accelerated aging tests establish that capacity loss per cycle exhibits nonlinear dependence on cycle depth. Following Xu et al.~\cite{xu2018_cycle_aging_market}, the stress function for a cycle of normalized depth $\delta \in [0,1]$ is:
\begin{equation}
\Phi(\delta) = a \cdot \delta^b
\label{eq:stress_function}
\end{equation}
where $a = 5.24 \times 10^{-4}$ and $b = 2.03$ are parameters for Li(NiMnCo)O$_2$ chemistry. The superlinear exponent $b > 1$ captures the critical nonlinearity: a single 50\% depth-of-discharge cycle causes more than twice the damage of two 25\% cycles.

Real-world battery operation involves irregular SOC trajectories with complex cycling patterns. The rainflow counting algorithm~\cite{endo1968_rainflow}, adopted from mechanical fatigue analysis, decomposes arbitrary profiles into equivalent cycles by identifying closed hysteresis loops, extracting cycle depths $\{\delta_i\}$ with counts $\{n_i\}$ where $n_i \in \{0.5, 1.0\}$ for half-cycles and full cycles. The cumulative capacity loss is converted to economic cost through:
\begin{equation}
C_{\text{cyc}}(\text{SOC}_{0:T}) = \frac{R}{\eta_{\text{dis}}} \sum_{i=1}^{N_{\text{cyc}}} n_i \cdot \Phi(\delta_i)
\label{eq:degradation_cost}
\end{equation}
where $R/\eta_{\text{dis}}$ accounts for reduced energy throughput as the battery degrades. This quantifies the marginal degradation cost incurred by any schedule, enabling direct economic comparison with energy arbitrage revenue. This pricing couples the slow aging process to the fast dispatch horizon, so that a schedule re-solved each day limits cumulative long-horizon wear without an explicit multi-year plan. Because $C_{\text{cyc}}$ is expressed in the same monetary units as the energy cost through the battery replacement cost, the objective is a true total cost rather than a weighted sum, and no relative weighting between arbitrage and aging need be tuned. The key computational challenge is that $C_{\text{cyc}}$ is a complex, nonconvex, and nondifferentiable function of the SOC trajectory due to the combinatorial nature of rainflow counting.

\subsection{Piecewise Linear Approximation Baseline}
\label{sec:pwl_method}
The Piecewise Linear approximation method~\cite{xu2018_cycle_aging_market} serves as our primary baseline. It addresses computational intractability by replacing the nonlinear, nondifferentiable $C_{\text{cyc}}$ with a convex linear surrogate. The battery capacity is partitioned into $J$ equal segments (typically $J=16$), each assigned a constant marginal degradation cost. The method introduces segment-specific variables: $p_{t,j}^{\text{ch}}, p_{t,j}^{\text{dis}} \geq 0$: power to/from segment $j$ at time $t$, $e_{t,j} \in [0, E_{\text{cap}}/J]$: energy in segment $j$ at time $t$, and
$v_t \in \{0,1\}$: discharge mode indicator. The optimization problem is formulated as:
\begin{equation}
\min \; \sum_t \left( p_t^{\text{imp}} \pi_t^{\text{buy}} - p_t^{\text{exp}} \pi_t^{\text{sell}} + \sum_{j=1}^{J} C_{\text{dis}}[j]\, p_{t,j}^{\text{dis}} \right)\Delta t
\label{eq:pwl_objective}
\end{equation}

\begin{subequations}
\label{eq:pwl_constraints}
\begingroup
\setlength{\jot}{8pt}
\begin{align}
\text{s.t.}\quad
& C_{\text{dis}}[j]
= \frac{R J}{\eta_{\text{dis}} E_{\text{cap}}}
\bigg[\Phi\!\left(\frac{j}{J}\right)-\Phi\!\left(\frac{j-1}{J}\right)\bigg],
&& \forall j \label{eq:pwl_marginal_cost} \\
& p_t^{\text{ch}} = \sum_{j=1}^{J} p_{t,j}^{\text{ch}}, \quad
  p_t^{\text{dis}} = \sum_{j=1}^{J} p_{t,j}^{\text{dis}},
&& \forall t \label{eq:pwl_agg_ch} \\
& e_{t+1,j} = e_{t,j} + \eta_{\text{ch}} p_{t,j}^{\text{ch}} \Delta t 
- \frac{p_{t,j}^{\text{dis}}}{\eta_{\text{dis}}} \Delta t \quad (\forall t,j),  \label{eq:pwl_dynamics} \\
& 0 \le e_{t,j} \le \frac{E_{\text{cap}}}{J},
\quad (\forall t,j), \label{eq:pwl_segment_limits} \\
& 0.1\,E_{\text{cap}}
\le \sum_{j=1}^{J} e_{t,j}
\le 0.9\,E_{\text{cap}},
&& \forall t \label{eq:pwl_soc_band} \\
& \sum_{j=1}^{J} e_{0,j} = \text{SOC}_0, \quad
  \sum_{j=1}^{J} e_{T,j} = \text{SOC}_T
\label{eq:pwl_terminal} \\
& p_t^{\text{ch}} \le P_{\text{max}}(1-v_t), \quad
  p_t^{\text{dis}} \le P_{\text{max}}v_t,
&& \forall t \label{eq:pwl_ch_mode} \\
& p_t^{\text{import}} - p_t^{\text{export}}
= d_t - s_t + p_t^{\text{ch}} - p_t^{\text{dis}},
&& \forall t \label{eq:pwl_balance}
\end{align}
\endgroup
\end{subequations}

The marginal degradation cost per kWh~\eqref{eq:pwl_marginal_cost} discretizes the stress function difference at segment boundaries. Convexity of $\Phi$ ($b > 1$) ensures $C_{\text{dis}}[1] < \cdots < C_{\text{dis}}[J]$, incentivizing shallow cycling. Initial segment energies are filled sequentially from $j=1$ upward until $\text{SOC}_0$ is distributed.

\subsection{Additional Benchmark Methods}
\label{sec:comparison_methods}
We benchmark our proposed method against two additional methods that are not part of the proposed framework. First, as a naive degradation-aware baseline, we replace the degradation term with a flat penalty on discharged energy, $C_{\text{thr}} = c_{\text{thr}} \sum_{t} p_t^{\text{dis}}\,\Delta t$, where $c_{\text{thr}}$~[\$/kWh] amortizes the battery replacement cost over its rated lifetime energy. All other constraints match the PWL formulation, giving a MILP solved by Gurobi. Being linear in throughput, this penalty charges shallow and deep cycles equally, whereas the true stress function $\Phi(\delta)\propto\delta^{2.03}$ penalizes deep cycles superlinearly. Second, since the exact MINLP with rainflow-based degradation is intractable for standard solvers, we employ differential evolution~\cite{storn1997differential} with penalty-based constraint handling to obtain near-optimal solutions on representative test cases.

\section{Methodology: Mixed-Integer Differentiable Predictive Control}
\label{sec:methods}

Figure~\ref{fig:architecture} provides an overview of the complete MI-DPC framework, and Algorithm~\ref{alg:l2o_training} details the architecture and training procedure. The architecture consists of three main components, (i) a neural policy network that encodes inputs and predicts control decisions through differentiable mode selection and power heads, (ii) forward simulation of battery dynamics with exact rainflow degradation evaluation, (iii) a differentiable layer enabling gradient-based training despite nondifferentiable rainflow counting. An optional convex projection layer applied only at inference guarantees constraint satisfaction. The framework integrates several established techniques from the literature, Gumbel-Softmax relaxation for discrete decisions\cite{jang2016gumbel_softmax}, curriculum learning for weight scheduling\cite{tang2025l2o_minlp}, and convex projection for feasibility, whose theoretical properties have been rigorously analyzed in prior works, into a cohesive architecture that, for the first time, enables gradient-based training on exact rainflow degradation.

\begin{algorithm}[t]
\caption{MI-DPC Training with Differentiable Rainflow Layer}
\label{alg:l2o_training}
\begin{algorithmic}[1]
\REQUIRE Training data $\mathcal{D} = \{(\mathbf{s}^{(i)}, \mathbf{d}^{(i)}, \boldsymbol{\pi}^{(i)})\}_{i=1}^N$ (solar, demand, prices)
\REQUIRE Total epochs $E$, curriculum schedule $w_{\text{deg}}(e)$
\ENSURE Trained policy network parameters $\theta^*$

\STATE Initialize network $\pi_\theta$ and optimizer

\FOR{epoch $e = 1$ to $E$}
    \STATE Update degradation weight: $w_{\text{deg}} \gets w_{\text{deg}}(e)$
    
    \FOR{batch $(\mathbf{s}, \mathbf{d}, \boldsymbol{\pi}) \sim \mathcal{D}$}
        
        \STATE \textbf{// Forward Pass}
        \STATE Encode inputs: $\mathbf{h} \gets \textsc{Encoder}_\theta(\mathbf{s}, \mathbf{d}, \boldsymbol{\pi})$
        \STATE Predict modes: $\mathbf{z}_t \gets \textsc{GumbelSoftmax}(\mathbf{h}_t)$ for $t = 0, \ldots, T-1$
        \STATE Predict power: $p_t^{\text{ch}}, p_t^{\text{dis}} \gets \textsc{PowerHead}_\theta(\mathbf{h}_t, \mathbf{z}_t)$
        \STATE Compute SOC: $\text{SOC}_t \gets f(\text{SOC}_{t-1}, p_t^{\text{ch}}, p_t^{\text{dis}})$ \COMMENT{Dynamics unroll}
        \STATE Compute grid exchange: $p_t^{\text{grid}} \gets d_t - s_t + p_t^{\text{ch}} - p_t^{\text{dis}}$
        
        \STATE \textbf{// Loss Computation}
        \STATE Economic cost: $\mathcal{L}_{\text{econ}} \gets \sum_t \text{cost}(p_t^{\text{grid}}, \boldsymbol{\pi}_t)$
        \STATE Degradation cost: $\mathcal{L}_{\text{deg}} \gets \textsc{ExactRainflow}(\text{SOC}_{0:T})$ \COMMENT{Exact rainflow in forward}
        \STATE Total loss: $\mathcal{L} \gets \mathcal{L}_{\text{econ}} + w_{\text{deg}} \mathcal{L}_{\text{deg}} + \mathcal{L}_{\text{penalties}}$
        
        \STATE \textbf{// Backward Pass (Differentiable Rainflow Layer)}
        \STATE Compute exact gradients at SOC extrema: $\nabla_{\text{exact}}$
        \STATE Compute proxy gradients between extrema: $\nabla_{\text{proxy}}$
        \STATE Combine: $\frac{\partial \mathcal{L}}{\partial \text{SOC}} \gets \alpha \nabla_{\text{exact}} + (1-\alpha) \nabla_{\text{proxy}}$
        \STATE Backpropagate through dynamics and network: $\nabla_\theta \mathcal{L}$
        
        \STATE Update parameters: $\theta \gets \textsc{Optimizer}(\theta, \nabla_\theta)$
        
    \ENDFOR
\ENDFOR

\RETURN Trained policy $\theta^*$
\end{algorithmic}
\end{algorithm}

\begin{figure*}[t]
\centering
\includegraphics[width=\linewidth]{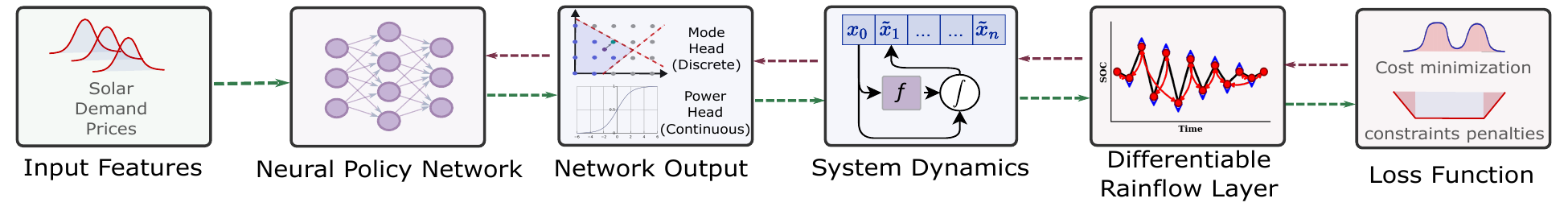}
\caption{MI-DPC architecture. Solar, demand, and price inputs are encoded by a neural policy with discrete mode and continuous power heads. SOC dynamics are unrolled deterministically through the differentiable rainflow layer, and gradients flow end-to-end from the combined economic-degradation loss back through the entire pipeline.}
\label{fig:architecture}
\end{figure*}

\subsection{Learning to Optimize}
\label{sec:l2o_background}
Learning-to-optimize replaces the repeated solution of a parametric optimization problem with the evaluation of a learned map~\cite{chen2022l2o_primer, amos2017optnet}. Rather than solving an optimization problem from scratch for every new set of inputs, one trains a neural policy once to map those inputs directly to the corresponding decisions, so that each new instance is solved by a single forward pass. Training may be supervised, using precomputed optimal solutions as labels, or self-supervised, directly minimizing the problem's own objective~\cite{donti2021dc3}; the latter needs no labels but requires that objective to be differentiable, which is precisely what the nondifferentiability of rainflow degradation obstructs and what the differentiable layer of this section restores. The approach is most attractive when a structurally identical problem is solved repeatedly with only its parameters changing. This is exactly the setting of battery dispatch, where the same program is solved each day, for each battery, with only the forecasts differing.

\subsection{Framework Overview}

The MI-DPC trains a neural network to directly predict near-optimal solutions from problem parameters, amortizing optimization effort through offline training to achieve constant-time $\mathcal{O}(1)$ inference. We reformulate the battery dispatch optimization~\eqref{eq:optimization_problem} as a parametric learning problem where the network predictions define a differentiable computational graph.

The battery dispatch problem is recast as:
\begin{equation}
\begin{aligned}
\min_{\theta} \quad \mathbb{E}_{\xi \sim \mathcal{D}} \Bigg[ 
&\sum_{t=0}^{T-1} \left(p_t^{\text{import}} \pi_t^{\text{buy}} - p_t^{\text{export}} \pi_t^{\text{sell}}\right) \Delta t \\
&+ w_{\text{deg}} \cdot C_{\text{cyc}}(\text{SOC}_{0:T}) \\
&+ \lambda_{\text{SOC}} \mathcal{L}_{\text{SOC}} + \lambda_{\text{term}} \mathcal{L}_{\text{term}} \Bigg]
\end{aligned}
\label{eq:l2o_formulation}
\end{equation}

\begin{subequations}
\begin{align}
\text{s.t.}\quad
&\text{SOC}_{t+1} = f(\text{SOC}_t, p_t^{\text{ch}}, p_t^{\text{dis}}, \Delta t), \quad \forall t \label{eq:l2o_dynamics}\\
&z_t = \pi_{\theta_1}(\xi_t), \quad z_t \in \{0, 1, 2\} \label{eq:l2o_mode}\\
&p_t^{\text{ch}}, p_t^{\text{dis}} = \pi_{\theta_2}(\xi_t, z_t) \label{eq:l2o_power}\\
&p_t^{\text{net}} = d_t - s_t + p_t^{\text{ch}} - p_t^{\text{dis}} \label{eq:l2o_net}\\
&p_t^{\text{import}} = [p_t^{\text{net}}]^+, \; p_t^{\text{export}} = [-p_t^{\text{net}}]^+ \label{eq:l2o_balance}\\
&\xi := [s_{0:T}, d_{0:T}, \pi_{0:T}^{\text{buy}}, \pi_{0:T}^{\text{sell}}] \label{eq:l2o_instance}
\end{align}
\end{subequations}

where $\mathcal{D}$ is the distribution of problem instances (solar generation $s_t$, demand $d_t$, electricity prices $\pi_t$), $f$ represents battery dynamics~\eqref{eq:soc_dyn}, $C_{\text{cyc}}$ is the exact rainflow degradation cost, and $\pi_{\theta_1}, \pi_{\theta_2}$ are neural network policies parameterized by $\theta$ that predict operating modes and power levels respectively. The dynamics~\eqref{eq:l2o_dynamics} and power balance~\eqref{eq:l2o_net}--\eqref{eq:l2o_balance} are enforced architecturally through deterministic computation, forming a differentiable computational graph that requires no relaxation for these constraints.

This formulation addresses three fundamental challenges, (i) handling discrete mode decisions through differentiable relaxations while maintaining architectural constraint satisfaction, (ii) computing gradients through the nondifferentiable rainflow function $C_{\text{cyc}}$ via a differentiable layer, and (iii) ensuring feasibility under distribution shift through optional convex projection. The following subsections detail our solutions organized by computational flow.

\subsection{Forward Pass}

The forward pass transforms problem parameters $\xi$ into a complete solution $\mathbf{x} = \{z_t, p_t^{\text{ch}}, p_t^{\text{dis}}, \text{SOC}_t, p_t^{\text{import}}, p_t^{\text{export}}\}$ through a sequence of neural and deterministic computations.

\subsubsection{Input Normalization and Encoding}
Raw problem parameters are normalized to facilitate learning. The input vector comprises solar generation and demand (normalized by maximum observed values), buy prices (normalized by infinity norm), and price spread (buy minus sell prices, similarly normalized) to highlight arbitrage opportunities, yielding $\mathbf{x}_{\text{in}} \in \mathbb{R}^{4T}$.

\subsubsection{Neural Network Architecture}

The policy network employs a deep encoder with dual output heads for discrete mode selection and continuous power prediction. The encoder processes normalized inputs through three fully-connected layers with layer normalization (maintaining consistent activation statistics) and dropout (regularization against overfitting):
\begin{align}
\mathbf{z}_1 &= \text{Dropout}_{p=0.1}\left(\text{ReLU}\left(\text{LayerNorm}(W_1 \mathbf{x}_{\text{in}} + \mathbf{b}_1)\right)\right) \\
\mathbf{z}_2 &= \text{Dropout}_{p=0.1}\left(\text{ReLU}\left(\text{LayerNorm}(W_2 \mathbf{z}_1 + \mathbf{b}_2)\right)\right) \\
\mathbf{h} &= \text{Dropout}_{p=0.1}\left(\text{ReLU}\left(\text{LayerNorm}(W_3 \mathbf{z}_2 + \mathbf{b}_3)\right)\right)
\end{align}

\subsubsection{Mode Selection via Gumbel-Softmax}

The binary constraints and mode exclusivity~\eqref{eq:mode_ex} are handled by representing the mode as a single categorical variable $m_t \in \{0, 1, 2\}$ corresponding to charge, discharge, and idle respectively\cite{jang2016gumbel_softmax}. Following~\cite{boldocky2025dpc, tang2025l2o_minlp, jang2016gumbel_softmax}, we apply the Gumbel-Softmax relaxation to maintain differentiability while producing discrete outputs.

The mode selection head produces logits $\mathbf{h}_t \in \mathbb{R}^3$ from the encoder output. Gumbel noise is sampled as $g_{t,k} \sim \text{Gumbel}(0, 1)$ for $k \in \{0, 1, 2\}$, and the soft (differentiable) mode vector is computed as:
\begin{equation}
\mathbf{z}^{\text{soft}}_t = \frac{\exp\left((\mathbf{h}_t + \mathbf{g}_t)/\tau\right)}{\sum_{k=0}^{2} \exp\left((h_{t,k} + g_{t,k})/\tau\right)}
\label{eq:gumbel_soft}
\end{equation}
where $\tau$ is the temperature parameter. For the forward pass, discrete one-hot vectors are obtained using the Straight-Through Estimator (STE)\cite{bengio2013stochastic_neurons}:
\begin{equation}
\mathbf{z}^{\text{hard}}_t = \text{one\_hot}\left(\arg\max_k \mathbf{z}^{\text{soft}}_{t,k}\right)
\label{eq:ste}
\end{equation}
while the backward pass uses gradients of $\mathbf{z}^{\text{soft}}_t$, ensuring exact mode exclusivity in the forward pass while maintaining gradient flow.

\subsubsection{Constraint-Aware Power Output Parameterization}

Power bounds $p_t^{\text{ch}}, p_t^{\text{dis}} \in [0, P_{\text{max}}]$ and mode coupling constraints~\eqref{eq:ch_couple}-\eqref{eq:dis_couple} are enforced architecturally\cite{drgona2024dpc_guarantees, tang2025l2o_minlp}. Raw power predictions $\mathbf{p}^{\text{raw}}_t \in \mathbb{R}^2$ from the power head are transformed as:
\begin{align}
p_t^{\text{ch}} &= z^{\text{hard}}_{t,0} \cdot \sigma\left(p^{\text{raw}}_{t,0}\right) \cdot P_{\text{max}} \label{eq:power_ch_arch}\\
p_t^{\text{dis}} &= z^{\text{hard}}_{t,1} \cdot \sigma\left(p^{\text{raw}}_{t,1}\right) \cdot P_{\text{max}} \label{eq:power_dis_arch}
\end{align}
where $\sigma(\cdot)$ is the sigmoid function. Multiplication by the mode indicator automatically satisfies the constraints, when idle mode is selected ($z^{\text{hard}}_{t,2} = 1$), both power outputs are zero by construction.

\subsubsection{Deterministic SOC Unrolling}

Rather than treating SOC as a learnable output, we compute it deterministically from power decisions through forward integration of the dynamics~\eqref{eq:soc_dyn}. This unrolling ensures exact satisfaction of the dynamics constraint and enables end-to-end gradient backpropagation through the entire trajectory via automatic differentiation \cite{li2024pdhg_unrolled_l2o}.

\subsubsection{Constraint-Aware Grid Exchange Parameterization}

To ensure exact satisfaction of the power balance constraint~\eqref{eq:power_bal}, the network predicts only battery power $(p_t^{\text{ch}}, p_t^{\text{dis}})$. Net grid power is computed deterministically as:
\begin{equation}
p_t^{\text{net}} = d_t - s_t + p_t^{\text{ch}} - p_t^{\text{dis}}
\label{eq:net_power}
\end{equation}

Grid import and export are then derived via rectified linear operations:
\begin{equation}
p_t^{\text{import}} = \text{ReLU}(p_t^{\text{net}}), \quad p_t^{\text{export}} = \text{ReLU}(-p_t^{\text{net}})
\label{eq:power_rectification}
\end{equation}

This enforces power balance by construction, eliminating the need for penalty terms and preserving differentiability through ReLU operations\cite{hendriks2020linearly_constrained_nn}.

\subsubsection{Rainflow Algorithm}

The final forward pass component computes exact degradation cost $C_{\text{cyc}}(\text{SOC}_{0:T})$ using rainflow cycle counting on the unrolled SOC trajectory. The algorithm identifies all closed hysteresis loops and extracts cycles $\{(\delta_i, n_i)\}$, recording peak/valley time indices $\{\tau_p, \tau_v\}$ and cycle depths $\{\delta_i\}$ for each. Half-cycles ($n = 0.5$) are included only for discharge events ($\text{SOC}_{\tau_s} > \text{SOC}_{\tau_e}$). The degradation cost is:
\begin{equation}
C_{\text{cyc}} = \frac{R}{\eta_{\text{dis}}} \sum_{i=1}^{N_{\text{cyc}}} n_i \cdot \Phi(\delta_i)
\end{equation}

This structural information, peak/valley indices and cycle depths, is saved to the Differentiable Rainflow Layer context for hybrid gradient computation in the backward pass (Section~\ref{sec:backward}).

\subsection{Loss Function}
The complete training loss aggregates the original optimization objective with differentiable constraint penalties:
\begin{equation}
\begin{aligned}
\mathcal{L}_{\text{total}} = & \mathcal{L}_{\text{econ}} + w_{\text{deg}}(e) \cdot \mathcal{L}_{\text{cyc}} + \lambda_{\text{SOC}} \mathcal{L}_{\text{SOC}} + \lambda_{\text{term}} \mathcal{L}_{\text{term}}
\end{aligned}
\label{eq:total_loss}
\end{equation}
where:
\begin{align}
\mathcal{L}_{\text{econ}} &= \mathbb{E}_{\text{batch}}\left[\sum_{t=0}^{T-1} \left(p_t^{\text{import}} \pi_t^{\text{buy}} - p_t^{\text{export}} \pi_t^{\text{sell}}\right) \Delta t\right] \label{eq:loss_econ}\\
\mathcal{L}_{\text{cyc}} &= \mathbb{E}_{\text{batch}}\left[C_{\text{cyc}}(\text{SOC}_{0:T})\right] \label{eq:loss_cyc}\\
\mathcal{L}_{\text{SOC}} &= \frac{1}{T+1}\sum_{t=0}^{T} \Big[\text{ReLU}(\text{SOC}_{\min} - \text{SOC}_t) \notag\\
&\quad\quad + \text{ReLU}(\text{SOC}_t - \text{SOC}_{\max})\Big] \label{eq:loss_soc_bounds}\\
\mathcal{L}_{\text{term}} &= \text{SmoothL1}(\text{SOC}_T, \text{SOC}_0) \label{eq:loss_terminal}
\end{align}

Here $w_{\text{deg}}(e)$ is the epoch-dependent degradation weight from curriculum learning (Section~\ref{sec:training}), and $\lambda_{\text{SOC}} = 20$, $\lambda_{\text{term}} = 10$ are fixed penalty weights. The economic cost $\mathcal{L}_{\text{econ}}$ minimizes grid energy costs, degradation cost $\mathcal{L}_{\text{cyc}}$ uses exact rainflow~\eqref{eq:degradation_cost} with the differentiable rainflow layer (Section~\ref{sec:backward}), SOC bound penalty $\mathcal{L}_{\text{SOC}}$ uses ReLU to penalize only violations, and terminal constraint $\mathcal{L}_{\text{term}}$ uses SmoothL1 loss\cite{girshick2015fast_rcnn} for stable gradient flow. This loss formulation enables end-to-end training since all components are differentiable.

\subsection{Backward Pass}
\label{sec:backward}

The backward pass computes gradients $\partial \mathcal{L} / \partial \theta$ through the entire computational graph. While most components use standard automatic differentiation, two aspects require specialized handling whiche are categorical mode variables and nondifferentiable rainflow counting. 

\subsubsection{Gradient Flow Through Categorical Variables}

Discrete mode selection $\mathbf{z}^{\text{hard}}_t$ is nondifferentiable, but gradients flow through the Gumbel-Softmax relaxation $\mathbf{z}^{\text{soft}}_t$ via the Straight-Through Estimator. During the backward pass:
\begin{equation}
\frac{\partial \mathcal{L}}{\partial \mathbf{h}_t} = \frac{\partial \mathcal{L}}{\partial \mathbf{z}^{\text{hard}}_t} \cdot \frac{\partial \mathbf{z}^{\text{soft}}_t}{\partial \mathbf{h}_t}
\end{equation}
where the first term treats discrete output as if it were continuous (straight-through), and the second term provides smooth gradients through the softmax operation. Temperature annealing progressively sharpens the distribution during training (Section~\ref{sec:training})\cite{tang2025l2o_minlp}.

\subsubsection{Differentiable Rainflow Layer via Hybrid Gradients}

The central technical contribution is a differentiable rainflow layer enabling gradient-based optimization despite the combinatorial, nondifferentiable nature of rainflow cycle counting. The key challenge is that while rainflow cost $C_{\text{cyc}}$ can be computed exactly in the forward pass, standard automatic differentiation fails because cycle identification relies on discrete stack operations that have no gradient.

We resolve this by computing a hybrid gradient $\partial \mathcal{L} / \partial \text{SOC}_t$ that combines two complementary components through a single mixing weight $\lambda \in [0,1]$:
\begin{equation}
\frac{\partial \mathcal{L}}{\partial \text{SOC}_t} = \lambda \, \nabla_{\text{exact}}(\text{SOC}_t) + (1-\lambda)\, \nabla_{\text{proxy}}(\text{SOC}_t)
\label{eq:hybrid_gradient}
\end{equation}
where $\lambda$ trades physical accuracy at cycle extrema against dense gradient coverage, and $\lambda = 0.5$ is adopted as a robust default (Section~\ref{sec:lambda_sensitivity}).

Exact Rainflow Gradient (Sparse).
Although the cycle identification process is nondifferentiable, once cycles are identified, the degradation cost is a differentiable function of the SOC values at cycle extrema. For each identified cycle $k$ with peak at timestep $\tau_p^k$ and valley at $\tau_v^k$, the cycle depth is $\delta_k = |\text{SOC}_{\tau_p^k} - \text{SOC}_{\tau_v^k}|/E_{\text{cap}}$. Applying the chain rule:
\begin{equation}
\nabla_{\text{exact}}(\text{SOC}_t) = \frac{\partial C_{\text{cyc}}}{\partial \text{SOC}_t} = \sum_{k: \{\tau_p^k, \tau_v^k\} \ni t} \frac{\partial}{\partial \text{SOC}_t}\left[\frac{R}{\eta_{\text{dis}}} \cdot a \cdot (\delta_k)^b\right]
\label{eq:exact_gradient_sum}
\end{equation}

For each cycle extremum, the gradient is:
\begin{equation}
\nabla_{\text{exact}}(\text{SOC}_t) = \begin{cases}
+\frac{R a b}{\eta_{\text{dis}} E_{\text{cap}}} (\delta_k)^{b-1} & \text{if } t = \tau_p^k \text{ (peak)} \\[4pt]
-\frac{R a b}{\eta_{\text{dis}} E_{\text{cap}}} (\delta_k)^{b-1} & \text{if } t = \tau_v^k \text{ (valley)} \\[4pt]
0 & \text{otherwise}
\end{cases}
\label{eq:exact_gradient}
\end{equation}

This sparse gradient is the exact mathematical derivative of the rainflow cost with respect to SOC at cycle extrema. It provides correct physics-based gradient direction where it matters most, at the points that define cycle depths. The exact gradient component is theoretically grounded in nonsmooth analysis. Under a no-ties condition ($\text{SOC}_t \neq \text{SOC}_{t+1}$ for all $t$), the rainflow cycle pairing is locally combinatorially rigid. For any trajectory perturbation smaller than half the minimum consecutive SOC difference, the set of detected cycles, their peak/valley indices, and their closure ordering remain unchanged. Within this neighborhood, $C_{\text{cyc}}$ reduces to a finite sum of smooth compositions of absolute differences, making it classically differentiable with gradient given exactly by~\eqref{eq:exact_gradient}. Since $C_{\text{cyc}}$ is locally Lipschitz (the number of cycles is bounded by $T/2$ and each cycle depth is Lipschitz in the trajectory), its Clarke generalized gradient is well-defined everywhere, and $\nabla_{\text{exact}}$ is a member of this subdifferential. The hybrid estimator~\eqref{eq:hybrid_gradient} therefore constitutes a bounded-variance Clarke subgradient oracle, ensuring that standard convergence guarantees for stochastic gradient descent on nonsmooth weakly convex objectives apply to the training procedure~\cite{davis2019stochastic}.

However, intermediate timesteps receive zero gradient, leading to potential gradient starvation during training. To address this limitation, we introduce a complementary dense component.

Smooth Proxy Gradient (Dense).
We define a differentiable proxy objective that approximates rainflow behavior:
\begin{equation}
C_{\text{proxy}} = \frac{R a}{\eta_{\text{dis}} E_{\text{cap}}^b} \sum_{t=0}^{T-1} \left(|\text{SOC}_{t+1} - \text{SOC}_t| + \epsilon\right)^b
\label{eq:proxy_cost}
\end{equation}
where $\epsilon = 10^{-6}$ ensures numerical stability near zero. This proxy penalizes the total magnitude of SOC variations with the same nonlinear power $b = 2.03$ as the rainflow stress function, providing a smooth surrogate that encourages shallow, infrequent cycling. The gradient via automatic differentiation is:
\begin{equation}
\begin{aligned}
\nabla_{\text{proxy}}(\text{SOC}_t)
&= \frac{Rab}{\eta_{\text{dis}} E_{\text{cap}}^b}
\Big[
|\Delta_{t-1}|^{b-1}\operatorname{sign}(\Delta_{t-1}) \\
&\quad -
|\Delta_t|^{b-1}\operatorname{sign}(\Delta_t)
\Big]
\end{aligned}
\label{eq:proxy_gradient}
\end{equation}
where $\Delta_t = \text{SOC}_{t+1} - \text{SOC}_t$. Unlike the exact gradient, this provides dense coverage with nonzero gradients at every timestep, enabling stable backpropagation through time.

The hybrid gradient~\eqref{eq:hybrid_gradient} combines exact physics at cycle extrema with a continuous optimization landscape everywhere, ensuring that (i) training optimizes the true rainflow degradation cost at critical points, (ii) gradient starvation is prevented through dense proxy coverage, and (iii) stable end-to-end training via backpropagation through time is maintained with physical fidelity. This combination reflects a bias--variance trade-off between the two components. The exact rainflow gradient is unbiased but sparse and high-variance, since it is supported only at the state-of-charge extrema where cycles close, whereas the proxy gradient is smooth and dense but biased. Mixing the two through $\lambda$ therefore trades the bias of the proxy against the variance of the exact gradient, yielding an informative descent direction everywhere while remaining faithful to the true degradation cost at the extrema. Figure~\ref{fig:hybrid_gradient} illustrates this decomposition. The exact gradient~(b) provides sparse but physics-accurate signals at cycle extrema, the proxy gradient~(c) offers dense coverage across all timesteps, the hybrid~(d) combines both components, and (e) shows that the hybrid scheme successfully minimizes the exact rainflow cost. The effect of the mixing weight $\lambda$ is examined in the sensitivity study of Section~\ref{sec:lambda_sensitivity}.

\begin{figure}[t]
\centering
\includegraphics[width=\columnwidth]{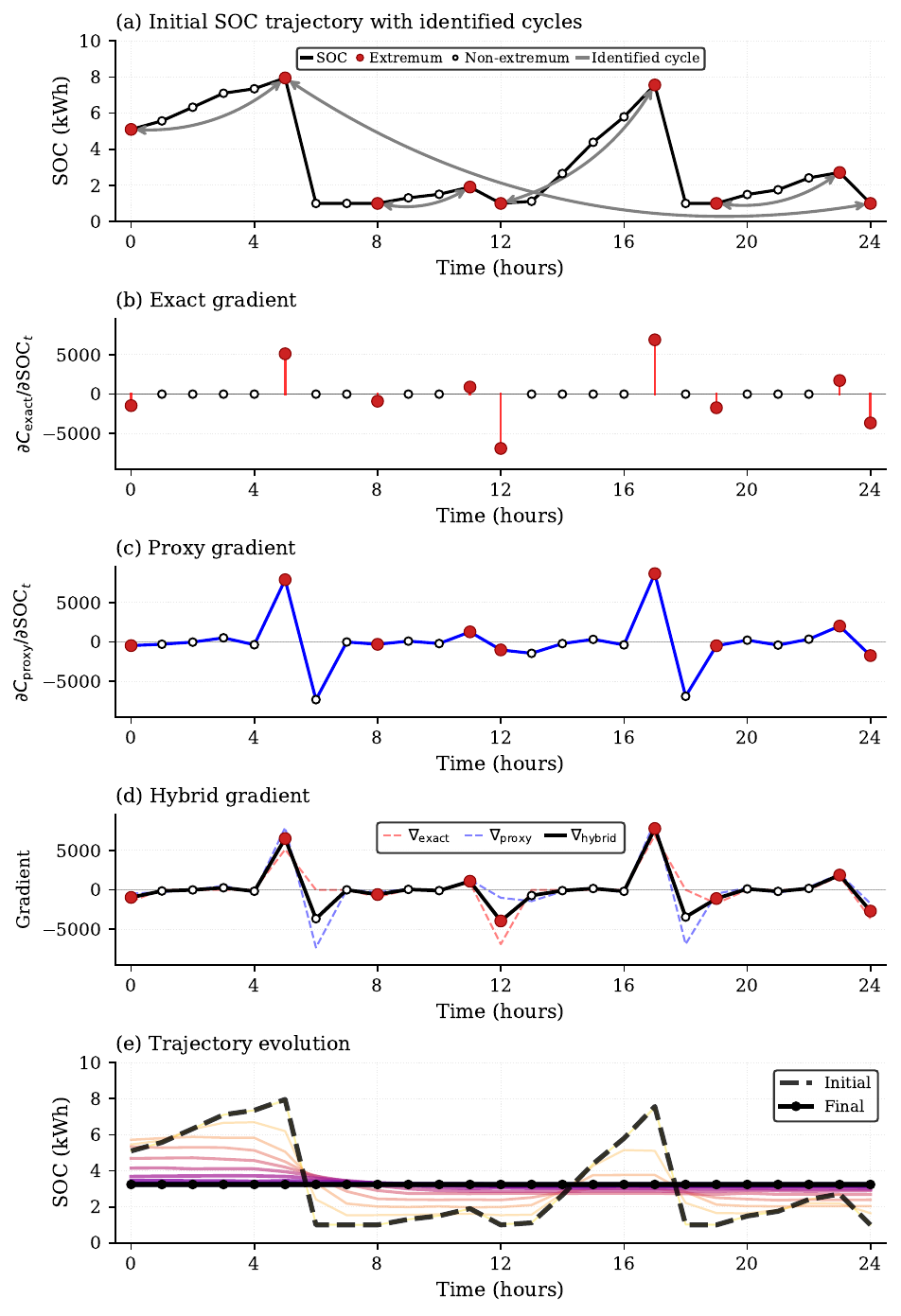}
\caption{Hybrid gradient approach for differentiable rainflow cycle counting.}
\label{fig:hybrid_gradient}
\end{figure}

\subsection{Training Strategies}
\label{sec:training}
We employ curriculum learning to balance competing objectives: the degradation cost weight decreases from $w_{\text{init}} = 2.0$ to $w_{\text{final}} = 1.0$ via square-root schedule over training epochs, teaching conservative cycling initially before allowing economic arbitrage exploitation. Simultaneously, Gumbel-Softmax temperature anneals linearly from $\tau = 5.0$ (uniform exploration) to $\tau = 0.1$ (deterministic selection)\cite{tang2025l2o_minlp}. To prevent early convergence to idle-only policies, the mode selection head bias is initialized as $[0.2, 0.2, -0.4]$ (charge, discharge, idle), favoring active modes\cite{bengio2021ml_combinatorial_optimization}. For robustness, inputs are perturbed with 5\% Gaussian noise during training to prevent overfitting and encourage generalization beyond specific profiles\cite{bertsimas2022online_mio}. We use OneCycleLR learning rate schedule\cite{smith2017superconvergence} ramping from $10^{-4}$ to $10^{-3}$ over the first 10\% of epochs, then cosine annealing to $10^{-5}$, combined with AdamW optimizer ($\lambda_{\text{wd}} = 10^{-4}$) and gradient clipping ($\|\nabla\|_2 \leq 1.0$) for stability.

\subsection{Safety Projection Filter}
\label{sec:feasibility_projection}

While the architectural constraints guarantee exact satisfaction of power balance~\eqref{eq:power_bal}, SOC dynamics~\eqref{eq:soc_dyn}, and mode-power coupling~\eqref{eq:ch_couple}--\eqref{eq:dis_couple}, the SOC bound constraints~\eqref{eq:soc_lim} are enforced only through soft penalties during training to preserve gradient flow. To ensure 100\% constraint satisfaction at inference, we introduce an optional convex projection filter applied only during deployment\cite{donti2021robust, wabersich2021predictive, nghiem2023physics, tang2025l2o_minlp}. Given neural network power predictions $(p_t^{\text{ch,NN}}, p_t^{\text{dis,NN}})$, the projection solves a quadratic program minimizing squared Euclidean distance to the network output subject to all hard constraints: 
\begin{equation}
\begin{aligned}
\min_{p^{\text{ch}}, p^{\text{dis}}, \text{SOC}} \quad & \sum_{t=0}^{T-1} \left[(p_t^{\text{ch}} - p_t^{\text{ch,NN}})^2 + (p_t^{\text{dis}} - p_t^{\text{dis,NN}})^2\right]
\end{aligned}
\label{eq:projection_objective}
\end{equation}
\begin{subequations}
\label{eq:projection_constraints}
\begin{align}
\text{s.t.}\quad
&\text{SOC}_{t+1} = \text{SOC}_t + \left(\eta_{\text{ch}} p_t^{\text{ch}} - \frac{p_t^{\text{dis}}}{\eta_{\text{dis}}}\right) \Delta t, && \forall t \label{eq:proj_dyn}\\
&\text{SOC}_{\min} \leq \text{SOC}_t \leq \text{SOC}_{\max}, && \forall t \label{eq:proj_soc}\\
&0 \leq p_t^{\text{ch}} \leq P_{\text{max}}, \; 0 \leq p_t^{\text{dis}} \leq P_{\text{max}}, && \forall t \label{eq:proj_power}\\
&\text{SOC}_0 = \text{SOC}_{\text{init}} \label{eq:proj_init}\\
& p_t^{\text{ch}} = 0 \quad \text{if } z_t^{\text{NN}} \neq \text{charge}, && \forall t \label{eq:proj_mode_ch}\\
& p_t^{\text{dis}} = 0 \quad \text{if } z_t^{\text{NN}} \neq \text{discharge}, && \forall t \label{eq:proj_mode_dis}
\end{align}
\end{subequations}
To maintain real-time performance, we employ selective projection. For each network output, we simulate the SOC trajectory and check for violations with tolerance $\epsilon = 10^{-4}$ kWh. Only infeasible samples are projected via~\eqref{eq:projection_objective}--\eqref{eq:projection_constraints}, minimizing deviation from the learned policy with minimal impact on solution quality. This approach combines soft constraint training for effective learning with hard constraint enforcement at deployment for guaranteed feasibility~\cite{nghiem2023physics}.

The projection~\eqref{eq:projection_objective}--\eqref{eq:projection_constraints} guarantees constraint satisfaction independent of network output quality. 

\begin{assumption}[Feasibility of Zero Action]
\label{ass:feasibility}
For any $\text{SOC}_0 \in [\text{SOC}_{\min}, \text{SOC}_{\max}]$, the zero-power action $p_t^{\text{ch}} = p_t^{\text{dis}} = 0$ is feasible for~\eqref{eq:projection_constraints}.
\end{assumption}
\begin{remark}
This holds by construction: zero power gives $\text{SOC}_{t+1}=\text{SOC}_t=\text{SOC}_0$ via~\eqref{eq:proj_dyn}, satisfying~\eqref{eq:proj_soc}--\eqref{eq:proj_init} and the mode-locking constraints~\eqref{eq:proj_mode_ch}--\eqref{eq:proj_mode_dis} for any $z_t^{\text{NN}} \in \Sigma$.
\end{remark}
\begin{theorem}[Well-Posedness and Feasibility]
\label{thm:wellposed}
Under Assumption~\ref{ass:feasibility}, the QP~\eqref{eq:projection_objective}--\eqref{eq:projection_constraints} admits a unique solution for any network output, and that solution satisfies $\text{SOC}_t \in [\text{SOC}_{\min}, \text{SOC}_{\max}]$ for all $t$.
\end{theorem}
\begin{proof}
The feasible set is nonempty (Assumption~\ref{ass:feasibility}), closed, and convex (linear equality and inequality constraints), and the objective~\eqref{eq:projection_objective} is strictly convex; a strictly convex function over such a set has a unique minimizer~\cite{bauschke2017convex}. Feasibility of that minimizer under~\eqref{eq:proj_soc} gives $\text{SOC}_t \in [\text{SOC}_{\min}, \text{SOC}_{\max}]$.
\end{proof}
\begin{remark}
Feasibility holds regardless of the quality of the network predictions or the optimality of the selected mode: the projection preserves the network's mode decisions and adjusts only the continuous power magnitudes, distributing the minimal $L_2$ correction across timesteps.
\end{remark}

\subsection{Extension to PWL Approximation}
\label{sec:l2o_pwl}
For ablation, we train an identical MI-DPC architecture on the PWL degradation formulation~(Section~\ref{sec:pwl_method}), where the linear surrogate is already differentiable and standard autograd replaces the hybrid rainflow layer. All other training strategies remain identical.

\subsection{Reinforcement Learning Baseline}
\label{sec:rl_baseline}
To contrast the proposed model-based approach with a model-free alternative, we include a reinforcement learning baseline that optimizes the same true objective (i.e., economic cost plus exact rainflow degradation) on the same problem, using Proximal Policy Optimization (PPO)~\cite{schulman2017ppo}. The agent observes the identical forecast seen by MI-DPC together with the realized SOC and time index, and its per-step reward is the economic cost, augmented at the terminal step by the exact rainflow degradation cost computed with the same cycle-inclusion logic used throughout; the return therefore equals the true cost.

Formally, the problem is cast as a finite-horizon Markov decision process $(\mathcal{S}, \mathcal{A}, P, r, T)$ over the $T$ dispatch intervals. The state $s_t = (\mathbf{f}, \text{SOC}_t, t) \in \mathcal{S}$ comprises the forecast $\mathbf{f}$ (solar, demand, buy price, and price spread), the realized state of charge $\text{SOC}_t$, and the time index. The action $a_t \in \mathcal{A}$ selects a mode $m_t \in \{\text{charge}, \text{discharge}, \text{idle}\}$ and a power level, from which the charge/discharge powers $(p_t^{\text{ch}}, p_t^{\text{dis}})$ follow. The transition $P$ is the deterministic SOC dynamics $\text{SOC}_{t+1} = \text{SOC}_t + (p_t^{\text{ch}}\eta_{\text{ch}} - p_t^{\text{dis}}/\eta_{\text{dis}})\,\Delta t$, and the reward is the negative per-step economic cost, with the exact rainflow degradation cost added at the terminal step so that the undiscounted return equals the true objective $-(C_{\text{energy}} + C_{\text{cyc}})$.

A key design choice is that the action space is mixed-integer rather than continuous. At each step the agent selects a discrete mode (charge, discharge, or idle) together with a power level, mirroring the binary mode structure that MI-DPC handles via the Gumbel--Softmax relaxation. A purely continuous-power baseline would solve an easier problem than MI-DPC and would not exercise the discrete mode-selection that is central to the dispatch task, making the comparison unfair. In contrast, the mixed-integer action space ensures that RL and MI-DPC face the same combinatorial decision, with the sole difference being model-free exploration of the modes versus the differentiable relaxation used by MI-DPC. SOC bounds are imposed as a soft penalty during training, and feasibility is reported as an outcome; the same QP safety filter used by MI-DPC can optionally be applied at inference for a like-for-like feasibility comparison.


\section{Experimental Results}
\label{sec:results}
We evaluate the MI-DPC framework through progressive generalization tests which are synthetic test scenarios, single-battery training generalization across the fleet, and fleet-wide training.

\subsection{Sensitivity to the Mixing Weight $\lambda$}
\label{sec:lambda_sensitivity}
As introduced in Section~III-E, the hybrid gradient reflects a bias--variance tradeoff whose optimum is interior to $(0,1)$ and trajectory-dependent. Figure~3 examines this over 50 trajectories. Fig. 3a shows that both endpoints converge slowest: for the pure proxy ($\lambda=0$) this is because its bias stalls descent near the optimum, while for the pure exact ($\lambda=1$) because its sparse gradient is uninformative between extrema. Meanwhile, interior values converge fastest. Fig. 3b sweeps $\lambda$ densely and reveals a broad flat basin over the interior in which the final cost is nearly constant, rising sharply only near either endpoint. Performance is therefore insensitive to the exact value of $\lambda$ within this basin so we adopt $\lambda=0.5$ as a robust default requiring no tuning, rather than as a claimed optimum.

\begin{figure}[t]
\centering
\includegraphics[width=0.9\columnwidth]{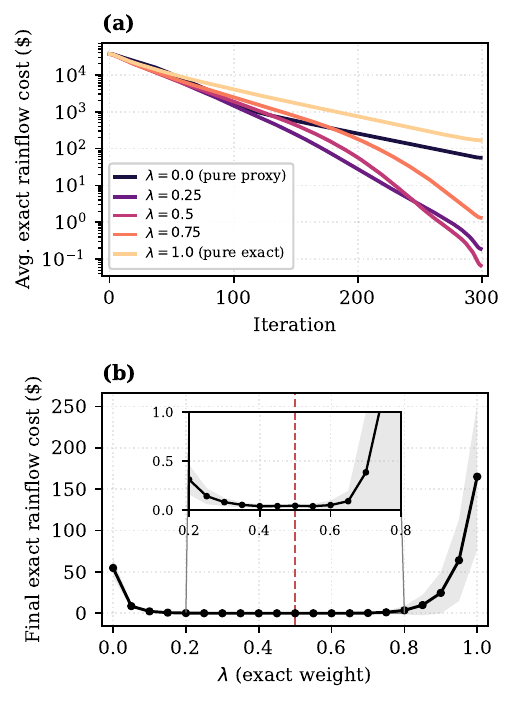}
\caption{Sensitivity of daily rainflow degradation cost to the mixing weight $\lambda$ over 50 trajectories. (a) Convergence of the exact rainflow cost. (b) Dense sweep showing dependence of cost on $\lambda$.}
\label{fig:lambda}
\end{figure}

\subsection{Experimental Setup}
\label{sec:experimental_setup}
We evaluate our approach on 10 residential batteries across three utility regions: 6 under SDG\&E (California)~\cite{sdge_nbt_tariff}, 1 under Xcel Energy (Colorado)~\cite{xcel_nem}, and 3 under APS (Arizona)~\cite{aps_rates}, each over 365 days, yielding 3{,}650 independent battery-day dispatch instances. These days span the full seasonal range across three climatically distinct regions, so that each instance presents a different combination of load, solar generation, and price profile under varying weather conditions. These territories span a range of time-of-use rate designs with markedly different degrees of price variation, which directly determines the arbitrage incentive available to storage. SDG\&E (TOU-DR1~\cite{sdge_nbt_tariff}) has the steepest structure, with a summer peak-to-off-peak differential of \$0.37/kWh; APS (R-TOU-E~\cite{aps_rates}) is intermediate at roughly \$0.22/kWh; and Xcel (RE-TOU~\cite{xcel_nem}) is the flattest of the three at \$0.21/kWh summer, narrowing to \$0.09/kWh in winter. Export compensation also differs. Under SDG\&E NEM~3.0, exports are credited at hourly avoided-cost rates well below retail, which sharpens the incentive for self-consumption through storage, whereas APS and Xcel likewise apply below-retail export compensation. The policy is thus exposed to tariff regimes ranging from strongly peaked to relatively flat, rather than a single rate design.

All batteries share identical specifications: 10 kWh capacity, 4.8 kW power limits, 92\% efficiency, SOC bounds [10\%, 90\%], and Li-NMC degradation parameters~\cite{xu2018_cycle_aging_market}. The degradation model is chemistry specific only through the two constants of the stress law $\Phi(\delta)=a\,\delta^{b}$. The differentiable rainflow layer, the architecture, and the training procedure are unchanged for any monotonic stress function, and adapting the framework to another chemistry amounts to recalibrating coefficients rather than modifying the method. We therefore evaluate the procedure on Li-NMC throughout and do not claim that the reported numerical trade-offs generalized unchanged to other chemistries, since $b$ enters as an exponent on cycle depth and reshapes the depth-cost relationship. Each battery has one year of real metered customer data (rooftop solar and household demand at 15-minute resolution); these are true residential traces and are therefore irregular and non-stationary by construction, not smoothed or synthetic signals. Evaluated over 365 days per battery, this yields $3{,}650$ distinct daily instances, each with its own load, solar generation, and price profile, so the results reflect a broad and heterogeneous set of operating conditions rather than a small sample of scenarios. Our training procedure combines real profiles with synthetic augmentation; the single-battery training uses one SDG\&E customer, while the fleet-wide training samples from all 10 batteries.

We compare against five methods: PWL (Section~\ref{sec:pwl_method}), the piecewise-linear degradation MILP solved via Gurobi and post-evaluated under exact rainflow; DE, differential evolution optimizing the exact rainflow objective directly (Section~\ref{sec:comparison_methods}); a linear throughput baseline (Section~\ref{sec:comparison_methods}) that penalizes discharged energy without cycle-depth structure; MI-DPC-PWL, an ablation with the proposed architecture trained on PWL gradients rather than exact rainflow (Section~\ref{sec:l2o_pwl}); and a reinforcement learning baseline (Section~\ref{sec:rl_baseline}), a PPO agent with a mixed-integer action space trained on the same objective. Metrics include the "performance gap vs. PWL (\%)", computational time and speedup (wall-clock vs. PWL), and feasibility. The "performance gap vs.\ PWL" is defined as the relative difference in total cost between the two schedules. We avoid the term optimality gap since the exact problem is a mixed-integer nonlinear program whose global optimum is not computable at this scale, so the PWL solution serves as a benchmark rather than ground truth.

All models were trained on a single NVIDIA Tesla T4 GPU (16~GB) hosted on a compute node with an Intel Xeon Gold 6248R CPU (3.00~GHz) and 1.5~TB of memory; training the MI-DPC policy required approximately 8.5 hours. This is a one-off offline cost, incurred once per fleet configuration and amortized over all subsequent instances the trained policy solves. The GPU is used only for this offline training; all reported solve times are per-instance inference, i.e. the time to produce one schedule for a single battery-day, with MI-DPC evaluated as a single batch-size-1 forward pass including the QP projection. These timings were measured on a separate machine with an Intel Core Ultra 7 155H CPU (16 cores) and 16~GB of memory, running both the Gurobi PWL benchmark (Gurobi 13.0) and MI-DPC inference on the same CPU, so the comparison is per-instance and hardware-matched.

\subsection{In-Distribution Performance with Synthetic Data}
\label{sec:in_dist}

We first evaluate on 10 test scenarios drawn from the training distribution but not seen during training. Table~\ref{tab:in_dist_performance} reports the full set of metrics. The proposed MI-DPC policy trained with exact rainflow achieves the best performance among all learning-based methods, with a mean performance gap of $+1.74\%$ relative to PWL while delivering a $229\times$ speedup ($41$~ms vs.\ $9.42$~s) with the safety projection enabled. Without projection, inference drops to $8$~ms for a $1{,}240\times$ speedup at identical solution quality, indicating that the learned policy already produces feasible schedules and that the projection serves as a guarantee rather than a corrective. MI-DPC attains $10/10$ feasibility in both cases, confirming that the soft-penalty training strategy teaches the network to respect SOC bounds implicitly.

The cost decomposition shows why the training objective matters. MI-DPC closely matches the cost structure of PWL in both its economic (\$29.02 vs.\ \$28.55) and degradation (\$0.62 vs.\ \$0.58) components, indicating that training on the true rainflow objective internalizes the correct economic and degradation trade-off. The two baselines that misprice degradation fail at opposite extremes. The throughput baseline, penalizing discharged energy with no cycle-depth structure, incurs the highest degradation cost (\$1.29) and lands at $+2.34\%$, while MI-DPC-PWL, an identical architecture trained on the linearized surrogate, over-suppresses cycling to reach the lowest degradation cost (\$0.15) but sacrifices arbitrage revenue and rises to $+8.14\%$. Holding the architecture fixed and changing only the training objective from the PWL surrogate to exact rainflow thus reduces the gap from $+8.14\%$ to $+1.74\%$, isolating the value of the exact rainflow layer.

The reinforcement learning baseline, which optimizes the same objective with the same mixed-integer decision structure but without a differentiable model, is not competitive, reaching $+17.90\%$ with projection and $+22.98\%$ without, and is feasible on only $6/10$ days before projection. This is consistent with the discrete mode-selection bottleneck, as the agent lacks a differentiable path from the terminal degradation cost to the mode decisions and cannot sustain the multi-interval charge and discharge patterns that arbitrage requires. The projection nonetheless restores RL feasibility to $10/10$, showing that the QP filter guarantees constraint satisfaction even for a policy that violates SOC bounds on its own. DE attains a $+3.43\%$ gap but requires roughly 59 minutes per day, confirming that derivative-free search, while a useful quality reference, is impractical for deployment. As none of these baselines is competitive across solution quality, feasibility, and computational cost together, the remainder of the analysis focuses on MI-DPC trained with exact rainflow.

\begin{table*}[t]
\centering
\caption{In-Distribution Performance: Comparison Across Methods (10 Test Days)}
\label{tab:in_dist_performance}
\footnotesize
\begin{tabular}{@{}lccccccc@{}}
\toprule
\textbf{Method} & \textbf{Economic} & \textbf{Degr.} & \textbf{Total} & \textbf{Gap} & \textbf{Feasible} & \textbf{Time} & \textbf{Speedup} \\
 & \textbf{Cost (\$)} & \textbf{Cost (\$)} & \textbf{Cost (\$)} & \textbf{(\%)} & \textbf{Days} & \textbf{(s)} & \textbf{vs. PWL} \\
\midrule
\multicolumn{7}{l}{\textit{Baseline Methods}} \\
PWL & $28.55$ & $0.58$ & $29.13$ & $0.00$\% & 10/10 & 9.42 & 1.0× \\
Throughput & $28.52$ & $1.29$ & $29.81$ & $+2.34$\% & 10/10 & 1.10 & 8.6× \\
DE ($n=3$) & $29.63$ & $0.50$ & $30.13$ & $+3.43$\% & 3/3 & 3559.1 & 0.003× \\
\midrule
\multicolumn{7}{l}{\textit{Reinforcement Learning (PPO)}} \\
RL (no proj.) & $35.62$ & $0.20$ & $35.82$ & $+22.98$\% & 6/10 & 0.071 & 132.5× \\
RL + Proj. & $34.18$ & $0.16$ & $34.34$ & $+17.90$\% & 10/10 & 0.132 & 71.4× \\
\midrule
\multicolumn{7}{l}{\textit{MI-DPC: Trained with PWL Approximation}} \\
MI-DPC-PWL (no proj.) & $31.35$ & $0.15$ & $31.50$ & $+8.14$\% & 10/10 & 0.0004 & 24{,}154× \\
MI-DPC-PWL + Proj. & $31.35$ & $0.15$ & $31.50$ & $+8.14$\% & 10/10 & 0.019 & 496× \\
\midrule
\multicolumn{7}{l}{\textit{MI-DPC: Trained with Exact Rainflow}} \\
$\mathbf{MI\text{-}DPC}$ \textbf{(no proj.)} & $\mathbf{29.02}$ & $\mathbf{0.62}$ & $\mathbf{29.64}$ & $\mathbf{+1.74}$\% & \textbf{10/10} & $\mathbf{0.008}$ & $\mathbf{1239.9}$× \\
$\mathbf{MI\text{-}DPC}$ \textbf{+ Proj.} & $\mathbf{29.01}$ & $\mathbf{0.62}$ & $\mathbf{29.63}$ & $\mathbf{+1.74}$\% & \textbf{10/10} & $\mathbf{0.041}$ & $\mathbf{228.7}$× \\
\bottomrule
\end{tabular}
\end{table*}

\subsection{Single-Battery Training Generalization}
\label{sec:single_battery_training}

To assess cross-customer and cross-utility generalization, we train the MI-DPC policy exclusively on one SDG\&E customer with synthetic augmentation, then evaluate on all 10 batteries in the fleet. This experiment reveals whether a policy learned from a single customer's data can transfer to different households with distinct solar and demand patterns and entirely different utility pricing structures.

\subsubsection{Performance on Training Battery}

Table~\ref{tab:annual_performance} shows annual performance on the training battery over 365 days. The MI-DPC policy achieves excellent results with +0.36\% performance gap (\$7691.03 vs. \$7663.38 baseline) and 208× speedup (21.06s vs. 4370.24s annually) when projection is enabled. Without projection, the gap remains nearly identical (+0.35\%), but 8 summer days (2.2\%) violate feasibility constraints, all resolved by projection with negligible cost impact (\$0.60 difference).

\begin{table*}[t]
\centering
\caption{Single-Battery Training Performance (365 Battery-Days)}
\label{tab:annual_performance}
\small
\begin{tabular}{lccccccc}
\toprule
\textbf{Method} & \textbf{Economic} & \textbf{Degr.} & \textbf{Total} & \textbf{Gap} & \textbf{Feasible} & \textbf{Time} & \textbf{Speedup} \\
 & \textbf{Cost (\$)} & \textbf{Cost (\$)} & \textbf{Cost (\$)} & \textbf{(\%)} & \textbf{Days} & \textbf{(s)} & \textbf{vs. PWL} \\
\midrule
\multicolumn{8}{l}{\textit{Baseline}} \\
No Battery & 8355.98 & 0.00 & 8355.98 & +9.04\% & 365/365 & --- & --- \\
\midrule
\multicolumn{8}{l}{\textit{Classical Optimization}} \\
PWL & 7502.97 & 160.41 & 7663.38 & 0.00\% & 365/365 & 4370.24 & 1.0× \\
\midrule
\multicolumn{8}{l}{\textit{MI-DPC: Trained on SDG\&E Only}} \\
MI-DPC (no proj.) & 7529.87 & 160.56 & 7690.43 & +0.35\% & 357/365 & 2.35 & 1859× \\
MI-DPC + Proj. & 7530.71 & 160.32 & 7691.03 & +0.36\% & 365/365 & 21.06 & 208× \\
\bottomrule
\end{tabular}
\end{table*}

\subsubsection{Generalization to Other Batteries}

We now evaluate the same single-battery-trained policy across all 10 batteries in the fleet, including the training battery. This tests whether a policy learned from one customer's data can generalize to the diverse operating conditions present across different households and utility regions. Table~\ref{tab:fleet_comparison} summarizes annual performance across the complete fleet.

The policy achieves an overall performance gap of +11.3\% across the 10-battery fleet while maintaining 100\% feasibility (with projection) and delivering 72× computational speedup. Without projection, the policy achieves nearly identical cost (+11.3\%) but violates feasibility on 457 days (12.5\%), demonstrating that projection is essential for cross-battery deployment. The computational advantage remains substantial where 485 seconds versus 35,011 seconds for PWL, or under 1 minute with projection disabled (52.7 seconds).

This represents reasonable generalization performance considering the policy was trained on only a single customer's data and must now operate across substantially different conditions: 6 SDG\&E customers with identical pricing but diverse household patterns, plus 4 batteries under entirely different utility tariffs with different peak period structures and export compensation ratios. Despite this variation, the policy maintains computational efficiency and demonstrates that cross-customer and cross-utility deployment is feasible with learned policies. The 11.3\% gap indicates room for improvement, motivating the next experiment, training on the full fleet to capture customer and utility diversity directly in the training data.

\subsection{Fleet-Wide Training}
\label{sec:fleet_training}

To address the generalization challenges observed with single-battery training, we now train the MI-DPC policy on the complete 3650 battery-days scenarios, incorporating data from all three utility regions and diverse customer profiles. This gives the network direct exposure to the full range of demand patterns, pricing structures, and operating conditions present in the deployment environment. Table~\ref{tab:fleet_comparison} reports performance across all 3650 battery-days. Fleet-wide training substantially improves performance over single-battery training, reducing the performance gap from $+11.3\%$ to $+3.33\%$, a $71\%$ reduction in gap magnitude. The policy achieves near-optimal performance across the entire fleet while maintaining 100\% feasibility even without projection, indicating that training on diverse data teaches the network to respect constraints naturally across different operating conditions.

\begin{table}[t]
\centering
\caption{Fleet Performance Comparison: Single-Battery vs. Fleet-Wide Training (3650 Total Battery-Days)}
\label{tab:fleet_comparison}
\small
\setlength{\tabcolsep}{3pt}
\begin{tabular}{lcccccc}
\toprule
\textbf{Method} & \textbf{Econ.} & \textbf{Degr.} & \textbf{Total} & \textbf{Gap} & \textbf{Feasible} & \textbf{Time (s)} \\
 & \textbf{(\$)} & \textbf{(\$)} & \textbf{(\$)} & \textbf{(\%)} & \textbf{(/3650)} & \textbf{(Speedup)} \\
\midrule
\multicolumn{7}{l}{\textit{Baseline}} \\
No Battery & 20{,}600 & 0 & 20{,}600 & --- & 3650 & --- \\
PWL & 14{,}691 & 1{,}310 & 16{,}000 & $0.00$ & 3650 & $\vcenter{\hbox{\shortstack{35011 \\ (1.0×)}}}$ \\
\midrule
\multicolumn{7}{l}{\textit{MI-DPC: Single-Battery Training}} \\
no proj. & 16{,}118 & 1{,}692 & 17{,}810 & $11.31$ & 3193 & $\vcenter{\hbox{\shortstack{52.7 \\ (664×)}}}$ \\
+ Proj. & 16{,}114 & 1{,}691 & 17{,}804 & $11.27$ & 3650 & $\vcenter{\hbox{\shortstack{485 \\ (72×)}}}$ \\
\midrule
\multicolumn{7}{l}{\textit{MI-DPC: Fleet-Wide Training}} \\
$\mathbf{no\ proj.}$ & $\mathbf{15{,}160}$ & $\mathbf{1{,}373}$ & $\mathbf{16{,}533}$ & $\mathbf{3.33}$ & \textbf{3650} & $\vcenter{\hbox{\shortstack{$\mathbf{21.9}$ \\ $\mathbf{(1598\times)}$}}}$ \\
$\mathbf{+\ Proj.}$ & $\mathbf{15{,}160}$ & $\mathbf{1{,}373}$ & $\mathbf{16{,}533}$ & $\mathbf{3.33}$ & \textbf{3650} & $\vcenter{\hbox{\shortstack{$\mathbf{62.0}$ \\ $\mathbf{(564\times)}$}}}$ \\
\bottomrule
\end{tabular}
\end{table}

The cost decomposition in Table~\ref{tab:fleet_comparison} clarifies the source of this gap. Of the $+3.33\%$ total, the economic component accounts for $2.93\%$ and degradation for only $0.39\%$; the proposed policy's degradation cost (\$1{,}373) is within $4.8\%$ of the PWL benchmark (\$1{,}310), a ratio of $1.05\times$. The gap therefore reflects a small reduction in arbitrage revenue rather than increased battery wear, confirming that the differentiable rainflow layer manages cycle aging effectively.

The computational advantage is pronounced. The policy attains a $564\times$ speedup with projection (62~s vs.\ 35{,}011~s for PWL across 3650 battery-days) and $1598\times$ without projection (22~s total), solving 3{,}650 distinct battery-day instances in under a minute of wall-clock time versus over nine hours for PWL. The identical costs with and without projection confirm that the learned policy generates naturally feasible solutions, with projection providing negligible cost benefit while guaranteeing constraint satisfaction.

These results show that incorporating multi-customer, multi-utility data during training is essential for consistent performance across heterogeneous fleets. The $3.33\%$ fleet-wide gap represents a practical deployment point, combining near-optimal operation with millisecond-scale inference that enables real-time control across all systems simultaneously.

\subsection{Feasibility Projection Analysis}
\label{sec:feasibility}
While the MI-DPC policy learns to respect SOC constraints through soft penalties during training, out-of-distribution scenarios can produce trajectories that violate physical bounds. For single-battery training (Section~\ref{sec:single_battery_training}), the policy achieved 3193/3650 feasible days (87.5\%), with 457 days violating SOC bounds by 0.31 kWh mean (2.3\% capacity), maximum 0.64 kWh. The projection filter resolved all violations, achieving 3650/3650 feasibility with negligible cost impact (\$17,810 without projection vs. \$17,804 with projection, \$6 difference total) and 51ms average solve time per infeasible day. Figure~\ref{fig:feasibility_projection} illustrates representative projection behavior. Overall, maintaining projection in deployment provides guaranteed constraint satisfaction at negligible computational and economic cost.

\begin{figure}[t]
\centering
\includegraphics[width=0.9\linewidth]{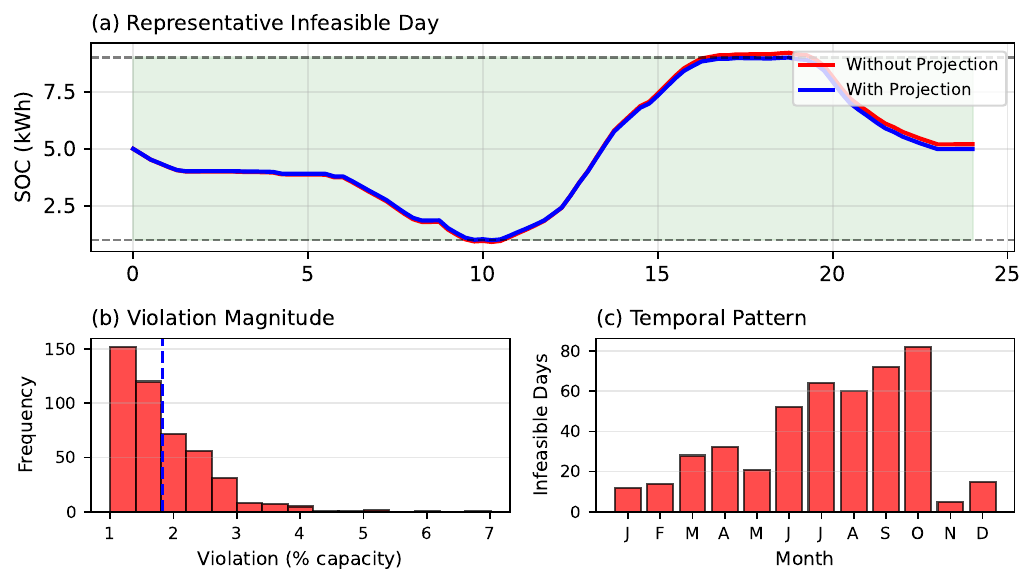}
\caption{Feasibility projection analysis. \textbf{(a)} Representative summer day showing SOC trajectory before and after projection. \textbf{(b)} Distribution of violation magnitudes across infeasible days. \textbf{(c)} Temporal pattern of violations.}
\label{fig:feasibility_projection}
\end{figure}

\subsection{Real-Time Fleet Aggregation}
\label{sec:aggregation}
The 3650 battery-days can be reinterpreted as a fleet-aggregation scenario requiring rapid dispatch decisions. Solving all instances with PWL requires roughly 9.7 hours sequentially, far exceeding any practical control window. MI-DPC solves each instance in a single per-battery forward pass on the same CPU, completing all 3650 in 62~s with projection (22~s without), well within real-time requirements. Because each schedule is produced independently, the approach scales naturally to larger fleets through parallel or batched execution.
\section{Conclusion}
\label{sec:conclusion}

This paper presents a Mixed-Integer Differentiable Predictive Control framework for battery energy storage dispatch that addresses the fundamental challenge of nondifferentiable rainflow cycle counting in degradation-aware optimization. The key technical contribution is a differentiable rainflow layer that enables end-to-end gradient-based training on true degradation physics rather than piecewise-linear approximations. Through careful architectural design, including Gumbel--Softmax mode selection, deterministic SOC unrolling, power balance rectification, and a convex safety projection with formal feasibility guarantees, the learned policy achieves feasible operation across diverse deployment conditions.

Experimental validation across 3{,}650 distinct battery-day instances, spanning a 10-battery fleet in three utility regions over 365 days, demonstrates practical viability. When trained and tested on a single battery, the policy achieves a $0.35\%$ performance gap relative to the PWL benchmark with over $200\times$ speedup and guaranteed feasibility via the safety projection filter. Deploying this policy to unseen households and utility regions degrades performance, as expected, but training on the full fleet resolves this generalization gap, achieving a $3.33\%$ gap across the entire fleet while maintaining full feasibility. The millisecond-scale inference makes coordination of large distributed battery fleets practical, a task infeasible with classical mixed-integer solvers.

Future work includes extending to receding-horizon control for intra-day re-optimization, incorporating online adaptation for non-stationary customer behavior, and integrating physics-based electrochemical degradation models beyond empirical stress functions.

\section*{Acknowledgements}
The authors gratefully acknowledge Arizona Public Service and sonnen, Inc., our utility and industry partners, for their collaboration and generous support in providing the data used in this work. This research was supported by the U.S. Department of Energy (DOE) under the BENEFIT program (Award No. DE-EE0010902), and by the DOE, Office of Science, ASCR program under the Scientific Discovery through Advanced Computing (SciDAC) Institute “LEADS: LEarning-Accelerated Domain Science”. 

\bibliographystyle{IEEEtran}
\bibliography{references}

\end{document}